\documentclass[12pt,letterpaper]{article}
\usepackage{amsmath,amssymb,amsfonts,amsthm,dsfont,mathtools,lipsum,etoolbox,mathabx}
\usepackage{booktabs}
\usepackage[flushleft]{threeparttable}
\usepackage{multirow}
\usepackage{setspace}
\usepackage{color}
\usepackage{endnotes}
\usepackage{enumerate}
\usepackage{float}
\usepackage{makecell}
\usepackage{afterpage}
\usepackage[bottom]{footmisc}
\usepackage[pdftex]{graphicx} 
\usepackage{epstopdf}
\usepackage{siunitx}
\usepackage{tikz}
\usepackage{indentfirst}
\usepackage{latexsym}
\usepackage{lscape}
\usepackage{parskip}
\usepackage{rotating}
\usepackage{caption}
\usepackage[normalem]{ulem}
\usepackage{pgfplots}
\usepackage{bm}
\usepackage{verbatim}
\usepackage{mathdots}
\usepackage{soul}
\usepackage{mwe}
\usepackage{arydshln}
\usepackage{xcolor}
\usepackage{subcaption}
\usepackage[scaled=0.85]{beramono}
\usepackage{placeins}
\usepackage[colorlinks=true,allcolors=blue]{hyperref}
\usepackage{cleveref}
\usepackage{chngcntr}
\usepackage{scalerel,stackengine}

\counterwithin{figure}{section}
\counterwithin{table}{section} 
\Crefname{subsection}{Subsection}{Subsections}

\usepackage[left=1.0in,right=1.0in,top=1.0in,bottom=1.0in]{geometry}
\usepackage[shrink=0,letterspace=0]{microtype}
\usepackage[authoryear, round,semicolon,sort&compress]{natbib}
\usepackage{mathptmx}
\usepackage[T1]{fontenc}

\allowdisplaybreaks[4]

\newtheorem{lma}{Lemma}[section]

\definecolor{mygreen}{rgb}{0.2, 0.7, 0.1}

\makeatletter
\newcommand{\leqnomode}{\tagsleft@true}
\newcommand{\reqnomode}{\tagsleft@false}
\makeatother

\newcommand\redsout{\bgroup\markoverwith{\textcolor{red}{\rule[0.5ex]{1pt}{1.4pt}}}\ULon}
\newcommand\bluesout{\bgroup\markoverwith{\textcolor{blue}{\rule[0.5ex]{1pt}{1.4pt}}}\ULon}

\newcommand{\Var}{\text{Var}}

\makeatletter
\newcommand{\VBig}{\bBigg@{2.7}}
\newcommand{\Vast}{\bBigg@{3.6}}
\newcommand{\VVast}{\bBigg@{5.5}}
\newcommand{\grp}{\mathrm{g}}   
\newcommand{\obs}{\mathrm{o}}   
\newcommand{\E}{\mathbb{E}}

\stackMath
\newcommand\reallywidehat[1]{%
	\savestack{\tmpbox}{\stretchto{%
			\scaleto{%
				\scalerel*[\widthof{\ensuremath{#1}}]{\kern-.6pt\bigwedge\kern-.6pt}%
				{\rule[-\textheight/2]{1ex}{\textheight}}
			}{\textheight}%
		}{0.5ex}}%
	\stackon[1pt]{#1}{\tmpbox}%
}

\makeatother

\newcommand{\zerodisplayskips}{%
	\setlength{\abovedisplayskip}{0.25cm}
	\setlength{\belowdisplayskip}{0.25cm}
	\setlength{\abovedisplayshortskip}{0.25cm}
	\setlength{\belowdisplayshortskip}{0.25cm}}
\appto{\normalsize}{\zerodisplayskips}
\appto{\small}{\zerodisplayskips}
\appto{\footnotesize}{\zerodisplayskips}

\newcommand{\eqfontsize}{\small}

\BeforeBeginEnvironment{equation*}{\begingroup\eqfontsize}
\AfterEndEnvironment{equation*}{\endgroup}

\BeforeBeginEnvironment{align*}{\begingroup\eqfontsize}
\AfterEndEnvironment{align*}{\endgroup}

\BeforeBeginEnvironment{equation}{\begingroup\eqfontsize}
\AfterEndEnvironment{equation}{\endgroup}

\BeforeBeginEnvironment{align}{\begingroup\eqfontsize}
\AfterEndEnvironment{align}{\endgroup}

\begin{document}	
	
\title{Unbiased Estimation of Central Moments in Unbalanced Two- and Three-Level Models}
\author{Dan Ben-Moshe\thanks{Department of Economics, Ben-Gurion University of the Negev. Email: dbmster@gmail.com}  
	\and David Genesove\thanks{Department of Economics, The Hebrew University of Jerusalem. Email: david.genesove@mail.huji.ac.il}}

\date{\today}

\maketitle
\begin{abstract}
	
This paper derives closed-form unbiased estimators of central moments in multilevel random-effects models with unbalanced group sizes. In a two-level model, we provide unbiased estimators for the second, third, and fourth central moments under both group-level and observation-level averaging. In a three-level model, we provide unbiased estimators for the second and third central moments.
	
	
\end{abstract}
	
\section{Introduction}

Multilevel (hierarchical) linear models are widely used in empirical economics, appearing in student--teacher--school assessments and neighborhood--city--region housing markets. While estimation of variance components in multilevel random-effects models is well established, including for unbalanced designs \citep[e.g.,][]{SearleCasellaMcCulloch1992VarComp,PattersonThompson1971REML}, there appear to be few closed-form results for finite-sample unbiased estimation of third and fourth central moments. Early ANOVA and cumulant work derives unbiased estimators for certain higher-order component cumulants in variance-component settings, often in balanced designs \citep{tukey1957third}. More recent work studies third and fourth moment estimation focusing on consistency and asymptotic efficiency rather than explicit finite-sample unbiased within--between decompositions under alternative averaging schemes \citep{wu2010orthogonality,wustutezhu2012efficient}.

In many economic applications, the third and fourth central moments of latent components (and their standardized counterparts, skewness and kurtosis) contain structural information about efficiency, selection, or risk. For example, work on earnings dynamics finds that income shocks exhibit substantial left-skewness and excess kurtosis \citep{guvenen2014nature}. More broadly, in consumption/income risk problems and insurance markets, higher moments of persistent and transitory shocks affect welfare, precautionary saving, and pricing of downside risk \citep{Kimball1990PrecautionarySaving,MenezesGeissTressler1980DownsideRisk}. In firm and establishment data, the distribution of growth and productivity shocks is often non-Gaussian, and higher moments are informative about investment, default risk, and aggregation \citep{BottazziSecchi2003TentShaped,Gabaix2011Granular}. The estimators in this paper are used in \citet{BenMosheGenesoveFrDe} in a two-level model and \citet{BenMosheGenesoveRegulation} in a three-level model.

This paper derives closed-form unbiased estimators of central moments in unbalanced two- and three-level models under two averaging schemes: group-level averaging (weighting groups equally) and observation-level averaging (weighting lowest-level units equally). In unbalanced panels, failing to account for unequal subgroup sizes can bias moment estimators for the latent components. The appropriate weighting depends on the object of interest: group-level averaging is natural when each group is meant to count equally, while observation-level averaging is natural when each lowest-level unit is meant to count equally.

Section~\ref{se:2level} presents the two-level results (second--fourth moments), Section~\ref{se:3level} presents the three-level results (second--third moments), and the Appendix provides the algebraic derivations.

\section{Two-level model} \label{se:2level}

Suppose we observe $y_{ij}$ in the two-level model
\[
y_{ij}=u_i+v_{ij},\quad i=1,\dots,n,\quad j=1,\dots,J_i,
\]
where $\{u_i\}_{i}$ and $\{v_{ij}\}_{i,j}$ are unobserved, mutually independent, and within each level i.i.d., with
$\E[u]=\E[v]=0$. We assume $n\geq 3$ and $J_i\geq 3$ for all~$i$. 
 Let $\mu_{ku}=\E[u^k]$ and $\mu_{kv}=\E[v^k]$ for $k=2,3,4$, and 
define sample means (superscripts $\grp$ and $\obs$ denote group-level and observation-level averaging) as follows,
\[
\bar y_i=\frac{1}{J_i}\sum_{j}y_{ij},\qquad
\bar y^{\grp}=\frac{1}{n}\sum_{i} \bar y_i,\qquad
\bar y^{\obs}=\frac{1}{N}\sum_{i,j} y_{ij},
\]
where $N=\sum_{i} J_i$ (we suppress summation limits where the range is clear). 
So $J_i$ is the size of group $i$, $n$ is the number of groups, and $N$ is the total number of observations. Then, $\bar y_i=u_i+\bar v_i$, $\bar y^{\grp}=\bar u^{\grp}+\bar v^{\grp}$, and $\bar y^{\obs}=\bar u^{\obs}+\bar v^{\obs}$.

Unbiased estimators of $\mu_{2v}$, $\mu_{3v}$, $\mu_{4v}$ are
\begin{align*}
	\widehat\mu_{2v}
	&=
	\frac{\sum_{i,j} (y_{ij}-\bar y_i)^2}{\sum_{i} (J_i-1)},
	\quad
	\widehat\mu_{3v}
	=
	\frac{\sum_{i,j} (y_{ij}-\bar y_i)^3}{\sum_{i} \frac{(J_i-1)(J_i-2)}{J_i}}, \quad
	\widehat\mu_{4v}
	=
	\frac{
		a_{22}^v \sum_{i,j} (y_{ij}-\bar y_i)^4
		-a_{12}^v \sum_{i}\sum_{1\le j<j'\le J_i} (y_{ij}-\bar y_i)^2 (y_{ij'}-\bar y_i)^2
	}{
		a_{11}^v a_{22}^v-a_{12}^v a_{21}^v
	}.
\end{align*}
The coefficients $a_{11}^v,a_{12}^v,a_{21}^v,a_{22}^v$ are given in Appendix~\ref{ap:2level}.

Unbiased estimators of $\mu_{2u}$, $\mu_{3u}$, and $\mu_{4u}$ under group-level averaging are
\begin{align*}
	\widehat\mu_{2u}^{\grp}
	&=
	\frac{1}{n-1}\sum_{i} (\bar y_i-\bar y^{\grp})^2
	-\frac{\widehat\mu_{2v}}{n}\sum_{i}\frac{1}{J_i},
	\qquad
	\widehat\mu_{3u}^{\grp}
	=
	\frac{n}{(n-1)(n-2)}\sum_{i} (\bar y_i-\bar y^{\grp})^3
	-\frac{\widehat\mu_{3v}}{n}\sum_{i}\frac{1}{J_i^2},\\
	\widehat\mu_{4u}^{\grp}
	&=
	\frac{
		a_{22}^{\grp}\Big(\sum_{i} (\bar y_i-\bar y^{\grp})^4-\widehat T_{4}^{\grp}\Big)
		-a_{12}^{\grp}\Big(\sum_{i\neq i'} (\bar y_i-\bar y^{\grp})^2(\bar y_{i'}-\bar y^{\grp})^2-\widehat T_{22}^{\grp}\Big)
	}{
		a_{11}^{\grp}a_{22}^{\grp}-a_{12}^{\grp}a_{21}^{\grp}
	}.
\end{align*}
The coefficients $a_{11}^{\grp},a_{12}^{\grp},a_{21}^{\grp},$ and $a_{22}^{\grp}$ and the adjustment terms $\widehat T_{4}^{\grp}$ and $\widehat T_{22}^{\grp}$ are given in Appendix~\ref{ap:2level}.

Unbiased estimators of $\mu_{2u}$, $\mu_{3u}$, and $\mu_{4u}$ under observation-level averaging are
\begin{align*}
	\widehat\mu_{2u}^{\obs}
	&=
	\frac{\sum_{i} (\bar y_i-\bar y^{\obs})^2
		-\widehat\mu_{2v}\sum_{i}\Big(\frac{1}{J_i}-\frac{1}{N}\Big)}
	{\sum_{i}\Big(1-\frac{2J_i}{N}+\frac{1}{N^2}\sum_{m=1}^n J_m^2\Big)},
	\quad
	\widehat\mu_{3u}^{\obs}
	=
	\frac{\sum_{i} (\bar y_i-\bar y^{\obs})^3
		-\widehat\mu_{3v}\sum_{i} \frac{(N-J_i)(N-2J_i)}{N^2J_i^2}}
	{\sum_{i}\Big(1-\frac{3J_i}{N}+\frac{3J_i^2}{N^2}-\frac{1}{N^3}\sum_{m=1}^n J_m^3\Big)},\\
	\widehat\mu_{4u}^{\obs}
	&=
	\frac{
		a_{22}^{\obs}\Big(\sum_{i} (\bar y_i-\bar y^{\obs})^4-\widehat T_{4}^{\obs}\Big)
		-a_{12}^{\obs}\Big(\sum_{i\neq i'} (\bar y_i-\bar y^{\obs})^2(\bar y_{i'}-\bar y^{\obs})^2-\widehat T_{22}^{\obs}\Big)
	}{
		a_{11}^{\obs}a_{22}^{\obs}-a_{12}^{\obs}a_{21}^{\obs}
	}.
\end{align*}
The coefficients $a_{11}^{\obs},a_{12}^{\obs},a_{21}^{\obs}$, and $a_{22}^{\obs}$ and the adjustment terms $\widehat T_{4}^{\obs}$ and $\widehat T_{22}^{\obs}$ are given in Appendix~\ref{ap:2level}.

Under balanced designs ($J_i=J$ for all~$i$), the group- and observation-level averaging estimators coincide. Under unbalanced designs, the difference between them is weighting: group-level averaging gives each group effect $u_i$ equal weight, while observation-level averaging weights group effects in proportion to $J_i$. Thus, the latter gives more influence to larger groups. This provides the main source of guidance to researchers: group-level averaging should be used when each group is meant to count equally, whereas observation-level averaging should be used when each lowest-level unit is meant to count equally. For example, if $u_i$ is a teacher effect and the goal is to describe the distribution of teacher quality across teachers, group-level averaging is natural. If instead the goal is to describe the distribution of teacher quality as experienced by a randomly sampled student, observation-level averaging is natural, since students are more likely to be assigned to larger classes. Group-level averaging may also be preferable when comparing moments across samples or periods, since otherwise changes in group sizes mechanically change the weights attached to the same group effects.

A second consideration is the variance of the estimator. To see the intuition, consider the second-moment estimators. Both use between-group variation in the sample means $\bar y_i$, but these means are measured with different precision: writing $\bar y_i=u_i+\bar v_i$, we have $\Var(\bar y_i)=\mu_{2u}+\mu_{2v}/J_i$, so smaller groups yield noisier sample means. Observation-level averaging therefore places more weight on larger groups, for which $\mu_{2v}/J_i$ is smaller. This suggests that, when group sizes are very unequal and $\mu_{2v}$ is large relative to $\mu_{2u}$, observation-level averaging may have smaller estimator variance. When $\mu_{2v}$ is small relative to $\mu_{2u}$, however, this variance advantage is attenuated. Because the two averaging schemes correspond to different objects of interest under unbalancedness, estimator variance is secondary to the choice of weighting.

\section{Three-level model}  \label{se:3level}

Suppose we observe $y_{ijk}$ in the three-level model,
\begin{align*}
	y_{ijk}=u_i+v_{ij}+w_{ijk},\quad i=1,\dots,n,\quad j=1,\dots,J_i,\quad k=1,\dots,K_{ij},
\end{align*}
where $\{u_i\}_{i}$, $\{v_{ij}\}_{i,j}$, and $\{w_{ijk}\}_{i,j,k}$ are unobserved, mutually independent, and within each level i.i.d., with
$\E[u]=\E[v]=\E[w]=0$. We assume $n\geq 3$, $J_i\geq 3$, and $K_{ij}\geq 3$ for all $(i,j)$. 
Let $\mu_{ku}=\E[u^k]$, $\mu_{kv}=\E[v^k]$, and $\mu_{kw}=\E[w^k]$ for $k=2,3$, and define sample means as follows,
	\begin{align*}
		\bar y_{ij}&=\frac{1}{K_{ij}}\sum_{k}y_{ijk},\quad
		\bar y_i^{\grp}=\frac{1}{J_i}\sum_{j}\bar y_{ij},\quad
		\bar y_i^{\obs}=\frac{1}{K_i}\sum_{j}K_{ij}\bar y_{ij},\quad
		\bar y^{\grp}=\frac{1}{n}\sum_{i} \bar y_i^{\grp},\quad
		\bar y^{\obs}=\frac{1}{N}\sum_{i,j,k} y_{ijk},\\
		\bar u^{\grp}&=\frac{1}{n}\sum_{i} u_i,\quad
		\bar u^{\obs}=\frac{1}{N}\sum_{i} K_i u_i,\quad
		\bar v_i^{\grp}=\frac{1}{J_i}\sum_{j} v_{ij},\quad
		\bar v_i^{\obs}=\frac{1}{K_i}\sum_{j}K_{ij}v_{ij},\quad
		\bar v^{\grp}=\frac{1}{n}\sum_{i} \bar v_i^{\grp},\quad
		\bar v^{\obs}=\frac{1}{N}\sum_{i,j}K_{ij}v_{ij},\\
		\bar w_{ij}&=\frac{1}{K_{ij}}\sum_{k} w_{ijk},\quad
		\bar w_i^{\grp}=\frac{1}{J_i}\sum_{j}\bar w_{ij},\quad
		\bar w_i^{\obs}=\frac{1}{K_i}\sum_{j}K_{ij}\bar w_{ij},\quad
		\bar w^{\grp}=\frac{1}{n}\sum_{i} \bar w_i^{\grp},\quad
		\bar w^{\obs}=\frac{1}{N}\sum_{i,j,k} w_{ijk}.
	\end{align*}
where $K_i=\sum_{j}K_{ij}$ and $N=\sum_{i} K_i$. So $K_{ij}$ is the size of group $(i,j)$, $J_i$ is the size of group $i$, $K_i$ is the number of groups in $i$, $n$ is the number of groups overall, and $N$ is the total number of observations. 
Then
$
	\bar y_{ij}=u_i+v_{ij}+\bar w_{ij}, $
	$\bar y_i^{\grp}=u_i+\bar v_i^{\grp}+\bar w_i^{\grp}, $
$	\bar y_i^{\obs}=u_i+\bar v_i^{\obs}+\bar w_i^{\obs}, $
$	\bar y^{\grp}=\bar u^{\grp}+\bar v^{\grp}+\bar w^{\grp}, $ and
$	\bar y^{\obs}=\bar u^{\obs}+\bar v^{\obs}+\bar w^{\obs}.$

Unbiased estimators of $\mu_{2w}$ and $\mu_{3w}$ are
\begin{align*}
	\widehat\mu_{2w}
	&=
	\frac{\sum_{i,j,k} (y_{ijk}-\bar y_{ij})^2}
	{\sum_{i,j}(K_{ij}-1)},
	\qquad
	\widehat\mu_{3w}
	=
	\frac{\sum_{i,j,k} (y_{ijk}-\bar y_{ij})^3}
	{\sum_{i,j}\frac{(K_{ij}-1)(K_{ij}-2)}{K_{ij}}}.
\end{align*}

Unbiased estimators of $\mu_{2v}$, $\mu_{3v}$,  $\mu_{2u}$, and $\mu_{3u}$ under group-level averaging are
\begin{align*}
	\widehat\mu_{2v}^{\grp}
	&=
	\frac{
		\sum_{i,j} \big(\bar y_{ij}-\bar y_i^{\grp}\big)^2
		-\widehat\mu_{2w}\sum_{i} \frac{J_i-1}{J_i}\sum_{j}\frac{1}{K_{ij}}
	}{
		\sum_{i} (J_i-1)
	},
	\qquad
	\widehat\mu_{3v}^{\grp}
	=
	\frac{
		\sum_{i,j} \big(\bar y_{ij}-\bar y_i^{\grp}\big)^3
		-\widehat\mu_{3w}\sum_{i} \frac{(J_i-1)(J_i-2)}{J_i^2}\sum_{j}\frac{1}{K_{ij}^2}
	}{
		\sum_{i} \frac{(J_i-1)(J_i-2)}{J_i}
	},\\
	\widehat\mu_{2u}^{\grp}
	&=
	\frac{1}{n-1}\sum_{i} \big(\bar y_i^{\grp}-\bar y^{\grp}\big)^2
	-\frac{\widehat\mu_{2v}^{\grp}}{n}\sum_{i}\frac{1}{J_i}
	-\frac{\widehat\mu_{2w}}{n}\sum_{i}\frac{1}{J_i^2}\sum_{j}\frac{1}{K_{ij}},
	\\
	\widehat\mu_{3u}^{\grp}
	&=
	\frac{n}{(n-1)(n-2)}\sum_{i} \big(\bar y_i^{\grp}-\bar y^{\grp}\big)^3
	-\frac{\widehat\mu_{3v}^{\grp}}{n}\sum_{i}\frac{1}{J_i^2}
	-\frac{\widehat\mu_{3w}}{n}\sum_{i}\frac{1}{J_i^3}\sum_{j}\frac{1}{K_{ij}^2}.
\end{align*}

Unbiased estimators of $\mu_{2v}$, $\mu_{3v}$,  $\mu_{2u}$, and $\mu_{3u}$  under observation-level averaging are
	\begin{align*}
		\widehat\mu_{2v}^{\obs}
		&=
		\frac{
			\sum_{i,j} \big(\bar y_{ij}-\bar y_i^{\obs}\big)^2
			-\widehat\mu_{2w}\sum_{i,j}\Big(\frac{1}{K_{ij}}-\frac{1}{K_i}\Big)
		}{
			\sum_{i,j}\Big(1-\frac{2K_{ij}}{K_i}+\frac{1}{K_i^2}\sum_{j'=1}^{J_i}K_{ij'}^2\Big)
		},
		\qquad
		\widehat\mu_{3v}^{\obs}
		=
		\frac{
			\sum_{i,j} \big(\bar y_{ij}-\bar y_i^{\obs}\big)^3
			-\widehat\mu_{3w}\sum_{i,j}\frac{(K_i-K_{ij})(K_i-2K_{ij})}{K_i^2K_{ij}^2}
		}{
			\sum_{i,j}\Big(1-\frac{3K_{ij}}{K_i}+\frac{3K_{ij}^2}{K_i^2}-\frac{1}{K_i^3}\sum_{j'=1}^{J_i}K_{ij'}^3\Big)
		},\\
		\widehat\mu_{2u}^{\obs}
		&=
		\frac{
			\sum_{i} \big(\bar y_i^{\obs}-\bar y^{\obs}\big)^2
			-\widehat\mu_{2v}^{\obs}\sum_{i}\Bigg[
			\Big(1-\frac{K_i}{N}\Big)^2\frac{1}{K_i^2}\sum_{j}K_{ij}^2
			+\frac{1}{N^2}\sum_{i'\neq i}\sum_{j=1}^{J_{i'}}K_{i'j}^2
			\Bigg]
			-\widehat\mu_{2w}\sum_{i}\Big(\frac{1}{K_i}-\frac{1}{N}\Big)
		}{
			\sum_{i}\Big(1-\frac{2K_i}{N}+\frac{1}{N^2}\sum_{i'=1}^n K_{i'}^2\Big)
		},
		\\
		\widehat\mu_{3u}^{\obs}
		&=
		\frac{
			\sum_{i} \big(\bar y_i^{\obs}-\bar y^{\obs}\big)^3
			-\widehat\mu_{3v}^{\obs}\sum_{i}\Bigg[
			\Big(1-\frac{K_i}{N}\Big)^3\frac{1}{K_i^3}\sum_{j}K_{ij}^3
			+\Big(-\frac{1}{N}\Big)^3\sum_{i'\neq i}\sum_{j=1}^{J_{i'}}K_{i'j}^3
			\Bigg]
			-\widehat\mu_{3w}\sum_{i} \frac{(N-K_i)(N-2K_i)}{N^2K_i^2}
		}{
			\sum_{i}\Big(1-\frac{3K_i}{N}+\frac{3K_i^2}{N^2}-\frac{1}{N^3}\sum_{i'=1}^n K_{i'}^3\Big)
		}.
	\end{align*}

\section{Conclusion} \label{se:con}

This paper derives closed-form unbiased estimators for the second, third, and fourth central moments in unbalanced two-level models, and for the second and third central moments in unbalanced three-level models. We provide these results under both group-level and observation-level averaging, allowing researchers to align estimation with their sampling design and population of interest.

The approach extends to richer multilevel designs via a three-step strategy: (i) write centered sample averages as weighted sums of latent components; (ii) apply Lemma~\ref{lma:powers} to express expectations of powers and cross-products of these averages in terms of latent moments; and (iii) solve the resulting linear system using sample analogs and previously derived lower-order estimators. Already at the fourth moment this requires systems involving cross-products, so closed-form expressions beyond the cases treated here are best derived case-by-case for the hierarchy and moments of interest.

	\bibliography{biblio}

\appendix

\section{Preliminaries}

\begin{lma}[Expected powers of weighted sums]\label{lma:powers}
	Let $x_1,\dots,x_m$ be i.i.d.\ with $\E[x]=0$ and $\mu_r=\E[x^r]$ for $r=2,3,4$.
	Let $S_m=\sum_{\ell=1}^m w_\ell x_\ell$. Then
	\begin{align*}
		\E[S_m^2]&=\mu_2\sum_{\ell=1}^m w_\ell^2,\qquad
		\E[S_m^3]=\mu_3\sum_{\ell=1}^m w_\ell^3,\qquad
		\E[S_m^4]=\mu_4\sum_{\ell=1}^m w_\ell^4
		+3\mu_2^2\Big(\Big(\sum_{\ell=1}^m w_\ell^2\Big)^2-\sum_{\ell=1}^m w_\ell^4\Big).
	\end{align*}
\end{lma}

\begin{proof}
	By expanding and using independence with $\E[x]=0$, any product term in which some $x_\ell$ appears to the first power has expectation zero. Therefore
	\begin{align*}
		\E[S_m^2]&=\sum_{\ell=1}^m w_\ell^2\E[x_\ell^2]=\mu_2\sum_{\ell=1}^m w_\ell^2,\qquad
		\E[S_m^3]=\sum_{\ell=1}^m w_\ell^3\E[x_\ell^3]=\mu_3\sum_{\ell=1}^m w_\ell^3,\\
		\E[S_m^4]&=\sum_{\ell=1}^m w_\ell^4\E[x_\ell^4]
		+6\sum_{\ell<r} w_\ell^2 w_r^2 \E[x_\ell^2]\E[x_r^2]
		=\mu_4\sum_{\ell=1}^m w_\ell^4
		+3\mu_2^2\Big(\Big(\sum_{\ell=1}^m w_\ell^2\Big)^2-\sum_{\ell=1}^m w_\ell^4\Big). \qedhere
	\end{align*}
\end{proof}

In what follows, unbiased estimators are obtained by replacing expectations with sample analogs and substituting unbiased estimators for unknown moments as they appear.

\section{Proofs for Two-Level Model} \label{ap:2level}

Using within-group variation,
\begin{align*}
	y_{ij}-\bar y_i
	&=v_{ij}-\bar v_i
	=\Big(1-\frac{1}{J_i}\Big)v_{ij}-\frac{1}{J_i}\sum_{j'\neq j} v_{ij'} .
\end{align*}
Applying Lemma~\ref{lma:powers} gives
\begin{align*}
	\E[(y_{ij}-\bar y_i)^2]
	&=\mu_{2v}\Big[\Big(1-\frac{1}{J_i}\Big)^2+(J_i-1)\Big(\frac{1}{J_i}\Big)^2\Big]
	=\mu_{2v}\frac{J_i-1}{J_i},\\
	\E[(y_{ij}-\bar y_i)^3]
	&=\mu_{3v}\Big[\Big(1-\frac{1}{J_i}\Big)^3+(J_i-1)\Big(-\frac{1}{J_i}\Big)^3\Big]
	=\mu_{3v}\frac{(J_i-1)(J_i-2)}{J_i^2},\\
	\E[(y_{ij}-\bar y_i)^4]
	&=\mu_{4v}\Big[\Big(1-\frac{1}{J_i}\Big)^4+(J_i-1)\Big(\frac{1}{J_i}\Big)^4\Big]\\
	&\quad
	+3\mu_{2v}^2\Bigg(
	\Big[\Big(1-\frac{1}{J_i}\Big)^2+(J_i-1)\Big(\frac{1}{J_i}\Big)^2\Big]^2
	-\Big[\Big(1-\frac{1}{J_i}\Big)^4+(J_i-1)\Big(\frac{1}{J_i}\Big)^4\Big]
	\Bigg)\\
	&=\mu_{4v}\frac{(J_i-1)(J_i^2-3J_i+3)}{J_i^3}
	+3\mu_{2v}^2\frac{(J_i-1)(2J_i-3)}{J_i^3},\\
		\E[(y_{ij}-\bar y_i)^2 (y_{ij'}-\bar y_i)^2]
	&=
	\mu_{4v}\Bigg[
	\Big(1-\frac{1}{J_i}\Big)^2\Big(\frac{1}{J_i}\Big)^2
	+\Big(\frac{1}{J_i}\Big)^2\Big(1-\frac{1}{J_i}\Big)^2
	+(J_i-2)\Big(\frac{1}{J_i}\Big)^4
	\Bigg] \\
	&+\mu_{2v}^2\Bigg\{
	\Bigg[\Big(1-\frac{1}{J_i}\Big)^2+(J_i-1)\Big(\frac{1}{J_i}\Big)^2\Bigg]^2
	+2\Bigg[(J_i-2)\Big(\frac{1}{J_i}\Big)^2 -\Big(1-\frac{1}{J_i}\Big)\Big(\frac{1}{J_i}\Big) 
	-\Big(\frac{1}{J_i}\Big)\Big(1-\frac{1}{J_i}\Big)
	\Bigg]^2 \\
	&\qquad
	-3\Bigg[
	\Big(1-\frac{1}{J_i}\Big)^2\Big(\frac{1}{J_i}\Big)^2
	+\Big(\frac{1}{J_i}\Big)^2\Big(1-\frac{1}{J_i}\Big)^2
	+(J_i-2)\Big(\frac{1}{J_i}\Big)^4
	\Bigg]
	\Bigg\}\\
	&=	\mu_{4v}\frac{2J_i-3}{J_i^3} 
	+\mu_{2v}^2\frac{J_i^3-2J_i^2-3J_i+9}{J_i^3}.
\end{align*}

Summing the fourth-moment identity over $(i,j)$ and the cross-product identity over  $(i, j<j')$ gives a $2\times2$ system. Inverting yields
\[
\binom{\mu_{4v}}{\mu_{2v}^2}
=
\begin{pmatrix}
	a_{11}^v & a_{12}^v\\
	a_{21}^v & a_{22}^v
\end{pmatrix}^{-1}
\binom{
	\E\!\left[\sum_{i=1}^n\sum_{j=1}^{J_i} (y_{ij}-\bar y_i)^4\right]
}{
	\E\!\left[\sum_{i=1}^n\sum_{1\le j<j'\le J_i} (y_{ij}-\bar y_i)^2 (y_{ij'}-\bar y_i)^2\right]
},
\]
where
	\begin{align*}
		a_{11}^v&=\sum_{i=1}^n \frac{(J_i-1)(J_i^2-3J_i+3)}{J_i^2},\qquad
		a_{12}^v=3\sum_{i=1}^n \frac{(J_i-1)(2J_i-3)}{J_i^2},\\
		a_{21}^v&=\frac12\sum_{i=1}^n \frac{(J_i-1)(2J_i-3)}{J_i^2},\qquad
		a_{22}^v=\frac12\sum_{i=1}^n \frac{(J_i-1)(J_i^3-2J_i^2-3J_i+9)}{J_i^2}.
	\end{align*}

Using between-group variation and group-averages,
\begin{align*}
	\bar y_i-\bar y^{\grp}
	&=\Big[u_i-\frac{1}{n}\sum_{i'=1}^n u_{i'}\Big]
	+\Big[\bar v_i-\frac{1}{n}\sum_{i'=1}^n \bar v_{i'}\Big]
	=\Big[\Big(1-\frac{1}{n}\Big)u_i-\frac{1}{n}\sum_{i'\neq i}u_{i'}\Big]
	+\Big[\Big(1-\frac{1}{n}\Big)\bar v_i-\frac{1}{n}\sum_{i'\neq i}\bar v_{i'}\Big].
\end{align*}
Applying Lemma~\ref{lma:powers} and using independence gives
\begin{align*}
	\E[(\bar y_i-\bar y^{\grp})^2]
	&=\mu_{2u}\Big[\Big(1-\frac{1}{n}\Big)^2+(n-1)\Big(\frac{1}{n}\Big)^2\Big]
	+\mu_{2v}\Big[\Big(1-\frac{1}{n}\Big)^2\frac{1}{J_i}
	+\Big(\frac{1}{n}\Big)^2\sum_{i'\neq i}\frac{1}{J_{i'}}\Big]\\
	&=\mu_{2u}\frac{n-1}{n}
	+\mu_{2v}\Big(\frac{1}{J_i}-\frac{2}{nJ_i}+\frac{1}{n^2}\sum_{i'=1}^n\frac{1}{J_{i'}}\Big),\\[4pt]
	\E[(\bar y_i-\bar y^{\grp})^3]
	&=\mu_{3u}\Big[\Big(1-\frac{1}{n}\Big)^3+(n-1)\Big(-\frac{1}{n}\Big)^3\Big]
	+\mu_{3v}\Big[\Big(1-\frac{1}{n}\Big)^3\frac{1}{J_i^2}
	+\Big(-\frac{1}{n}\Big)^3\sum_{i'\neq i}\frac{1}{J_{i'}^2}\Big]\\
	&=\mu_{3u}\frac{(n-1)(n-2)}{n^2}
	+\mu_{3v}\Big(\frac{1}{J_i^2}-\frac{3}{nJ_i^2}+\frac{3}{n^2J_i^2}-\frac{1}{n^3}\sum_{i'=1}^n\frac{1}{J_{i'}^2}\Big),\\[4pt]
	\E[(\bar y_i-\bar y^{\grp})^4]
	&=\E[(u_i-\bar u^{\grp})^4]+6\,\E[(u_i-\bar u^{\grp})^2]\E[(\bar v_i-\bar v^{\grp})^2]+\E[(\bar v_i-\bar v^{\grp})^4]\\
	&=\mu_{4u}\frac{(n-1)(n^2-3n+3)}{n^3}
	+3\mu_{2u}^2\frac{(n-1)(2n-3)}{n^3}\\
	&\quad
	+6\mu_{2u}\mu_{2v}\frac{n-1}{n}\Big(\frac{1}{J_i}-\frac{2}{nJ_i}+\frac{1}{n^2}\sum_{i'=1}^n\frac{1}{J_{i'}}\Big)
	+\mu_{4v}\Big[\Big(1-\frac{1}{n}\Big)^4\frac{1}{J_i^3}
	+\Big(\frac{1}{n}\Big)^4\sum_{i'\neq i}\frac{1}{J_{i'}^3}\Big]\\
	&\quad
	+3\mu_{2v}^2\Bigg(
	\Big[\Big(1-\frac{1}{n}\Big)^2\frac{1}{J_i}
	+\Big(\frac{1}{n}\Big)^2\sum_{i'\neq i}\frac{1}{J_{i'}}\Big]^2
	-\Big[\Big(1-\frac{1}{n}\Big)^4\frac{1}{J_i^3}
	+\Big(\frac{1}{n}\Big)^4\sum_{i'\neq i}\frac{1}{J_{i'}^3}\Big]
	\Bigg),\\[4pt]
	\E[(\bar y_i-\bar y^{\grp})^2(\bar y_{i'}-\bar y^{\grp})^2]
	&=
	\mu_{4u}\frac{2n-3}{n^3}
	+\mu_{2u}^2\frac{n^3+6n^2-11n+9}{n^3} 	-4\mu_{2u}\mu_{2v}\frac{n-1}{n^2}\Bigg(
	\frac{1}{n^2}\sum_{m=1}^n\frac{1}{J_m}-\frac{1}{nJ_i}-\frac{1}{nJ_{i'}}
	\Bigg)\\
	&\quad
	+\mu_{2u}\mu_{2v}\frac{n-1}{n}\Bigg[
	\Big(\frac{1}{J_i}-\frac{2}{nJ_i}+\frac{1}{n^2}\sum_{m=1}^n\frac{1}{J_m}\Big)
	+\Big(\frac{1}{J_{i'}}-\frac{2}{nJ_{i'}}+\frac{1}{n^2}\sum_{m=1}^n\frac{1}{J_m}\Big)
	\Bigg]\\
	&\quad
	+\mu_{2v}^2\Bigg[
	\Big(\frac{1}{J_i}-\frac{2}{nJ_i}+\frac{1}{n^2}\sum_{m=1}^n\frac{1}{J_m}\Big)
	\Big(\frac{1}{J_{i'}}-\frac{2}{nJ_{i'}}+\frac{1}{n^2}\sum_{m=1}^n\frac{1}{J_m}\Big)
	+2\Big(\frac{1}{n^2}\sum_{m=1}^n\frac{1}{J_m}-\frac{1}{nJ_i}-\frac{1}{nJ_{i'}}\Big)^2
	\Bigg]\\
	&\quad
	+\big(\mu_{4v}-3\mu_{2v}^2\big)\frac{1}{n^4}\Bigg[
	\sum_{m=1}^n\frac{1}{J_m^3}+(n^2-2n)\Big(\frac{1}{J_i^3}+\frac{1}{J_{i'}^3}\Big)
	\Bigg].
\end{align*}

Summing the fourth-moment identity over $i$ and the cross-product identity over $i\neq i'$ gives a $2\times2$ system. Inverting yields
\[
\binom{\mu_{4u}}{\mu_{2u}^2}
=
\begin{pmatrix}
	a_{11}^{\grp} & a_{12}^{\grp}\\
	a_{21}^{\grp} & a_{22}^{\grp}
\end{pmatrix}^{-1}
\binom{
	\E\!\left[\sum_{i=1}^n (\bar y_i-\bar y^{\grp})^4 - T_4^{\grp}\right]
}{
	\E\!\left[\sum_{i\neq i'} (\bar y_i-\bar y^{\grp})^2(\bar y_{i'}-\bar y^{\grp})^2 - T_{22}^{\grp}\right]
},
\]
\vspace{-0.035cm}
where
	\begin{align*}
	a_{11}^{\grp}&=\frac{(n-1)(n^2-3n+3)}{n^2},\quad
	a_{12}^{\grp}=\frac{3(n-1)(2n-3)}{n^2},\quad
	a_{21}^{\grp}=\frac{(n-1)(2n-3)}{n^2},\quad
	a_{22}^{\grp}=\frac{(n-1)(n^3+6n^2-11n+9)}{n^2},\\[6pt]
	T_4^{\grp}
	&=
	6\mu_{2u}\mu_{2v}\,\frac{(n-1)^2}{n^2}\sum_{i=1}^n\frac{1}{J_i}
	+\mu_{4v}\,\frac{(n-1)(n^2-3n+3)}{n^3}\sum_{i=1}^n\frac{1}{J_i^3}\\
	& \quad +3\mu_{2v}^2\Bigg(
	\frac{(n-2)^2}{n^2}\sum_{i=1}^n\frac{1}{J_i^2}
	+\frac{2n-3}{n^3}\Big(\sum_{i=1}^n\frac{1}{J_i}\Big)^2
	-\frac{(n-1)(n^2-3n+3)}{n^3}\sum_{i=1}^n\frac{1}{J_i^3}
	\Bigg),\\
	T_{22}^{\grp}
	&=
	2\mu_{2u}\mu_{2v}\,\frac{(n-1)\big((n-1)^2+2\big)}{n^2}\sum_{i=1}^n\frac{1}{J_i}
	+\mu_{4v}\,\frac{(n-1)(2n-3)}{n^3}\sum_{i=1}^n\frac{1}{J_i^3}\\
	&\quad+\mu_{2v}^2\Bigg(
	\frac{n^3-2n^2-3n+9}{n^3}\Big(\sum_{i=1}^n\frac{1}{J_i}\Big)^2
	+\frac{(n-2)(6-n)}{n^2}\sum_{i=1}^n\frac{1}{J_i^2}
	-3\,\frac{(n-1)(2n-3)}{n^3}\sum_{i=1}^n\frac{1}{J_i^3}
	\Bigg).
	\end{align*}

Similarly, using between-group variation and observation-averages,
\begin{align*}
	\bar y_i-\bar y^{\obs}
	&=\Big[\Big(1-\frac{J_i}{N}\Big)u_i-\frac{1}{N}\sum_{i'\neq i}J_{i'}u_{i'}\Big]
	+\Big[\Big(\frac{1}{J_i}-\frac{1}{N}\Big)\sum_{j=1}^{J_i}v_{ij}
	-\frac{1}{N}\sum_{i'\neq i}\sum_{j=1}^{J_{i'}}v_{i'j}\Big].
\end{align*}
Applying Lemma~\ref{lma:powers} and using independence gives
\begin{align*}
	\E[(\bar y_i-\bar y^{\obs})^2]
	&=\mu_{2u}\Big[\Big(1-\frac{J_i}{N}\Big)^2+\frac{1}{N^2}\sum_{i'\neq i}J_{i'}^2\Big]
	+\mu_{2v}\Big[J_i\Big(\frac{1}{J_i}-\frac{1}{N}\Big)^2+(N-J_i)\Big(\frac{1}{N}\Big)^2\Big]\\
	&=\mu_{2u}\Big(1-\frac{2J_i}{N}+\frac{1}{N^2}\sum_{m=1}^n J_m^2\Big)
	+\mu_{2v}\Big(\frac{1}{J_i}-\frac{1}{N}\Big),\\
	\E[(\bar y_i-\bar y^{\obs})^3]
	&=\mu_{3u}\Big[\Big(1-\frac{J_i}{N}\Big)^3-\frac{1}{N^3}\sum_{i'\neq i}J_{i'}^3\Big]
	+\mu_{3v}\Big[J_i\Big(\frac{1}{J_i}-\frac{1}{N}\Big)^3-(N-J_i)\Big(\frac{1}{N}\Big)^3\Big]\\
	&=\mu_{3u}\Big(1-\frac{3J_i}{N}+\frac{3J_i^2}{N^2}-\frac{1}{N^3}\sum_{m=1}^n J_m^3\Big)
	+\mu_{3v}\frac{(N-J_i)(N-2J_i)}{N^2J_i^2},\\
	\E[(\bar y_i-\bar y^{\obs})^4]
	&=\E[(u_i-\bar u^{\obs})^4]+6\,\E[(u_i-\bar u^{\obs})^2]\E[(\bar v_i-\bar v^{\obs})^2]+\E[(\bar v_i-\bar v^{\obs})^4]\\
	&=\mu_{4u}\Big[\Big(1-\frac{J_i}{N}\Big)^4+\frac{1}{N^4}\sum_{i'\neq i}J_{i'}^4\Big]
	+3\mu_{2u}^2\Bigg(
	\Big[\Big(1-\frac{J_i}{N}\Big)^2+\frac{1}{N^2}\sum_{i'\neq i}J_{i'}^2\Big]^2
	-\Big[\Big(1-\frac{J_i}{N}\Big)^4+\frac{1}{N^4}\sum_{i'\neq i}J_{i'}^4\Big]
	\Bigg)\\
	&\quad
	+6\mu_{2u}\mu_{2v}\Big[\Big(1-\frac{J_i}{N}\Big)^2+\frac{1}{N^2}\sum_{i'\neq i}J_{i'}^2\Big]
	\Big(\frac{1}{J_i}-\frac{1}{N}\Big)
	+\mu_{4v}\Big[J_i\Big(\frac{1}{J_i}-\frac{1}{N}\Big)^4+(N-J_i)\Big(\frac{1}{N}\Big)^4\Big]\\
	&\quad
	+3\mu_{2v}^2\Bigg(
	\Big(\frac{1}{J_i}-\frac{1}{N}\Big)^2
	-\Big[J_i\Big(\frac{1}{J_i}-\frac{1}{N}\Big)^4+(N-J_i)\Big(\frac{1}{N}\Big)^4\Big]
	\Bigg),\\[4pt]
	\E\!\left[(\bar y_i-\bar y^{\obs})^2(\bar y_{i'}-\bar y^{\obs})^2\right]
	&=
	\mu_{4u}\Bigg[
	\frac{J_i^2}{N^2}\Big(1-\frac{J_i}{N}\Big)^2
	+\frac{J_{i'}^2}{N^2}\Big(1-\frac{J_{i'}}{N}\Big)^2
	+\frac{1}{N^4}\Big(\sum_{m=1}^n J_m^4-J_i^4-J_{i'}^4\Big)
	\Bigg]\\
	&\quad
	+\mu_{2u}^2\Bigg(
	\Big(1-\frac{2J_i}{N}+\frac{1}{N^2}\sum_{m=1}^n J_m^2\Big)
	\Big(1-\frac{2J_{i'}}{N}+\frac{1}{N^2}\sum_{m=1}^n J_m^2\Big)
	+2\Big(\frac{1}{N^2}\sum_{m=1}^n J_m^2-\frac{J_i+J_{i'}}{N}\Big)^2\\
	&\qquad\qquad
	-3\Bigg[
	\frac{J_i^2}{N^2}\Big(1-\frac{J_i}{N}\Big)^2
	+\frac{J_{i'}^2}{N^2}\Big(1-\frac{J_{i'}}{N}\Big)^2
	+\frac{1}{N^4}\Big(\sum_{m=1}^n J_m^4-J_i^4-J_{i'}^4\Big)
	\Bigg]
	\Bigg)\\
	&\quad
	+\mu_{2u}\mu_{2v}\Bigg[
	\Big(1-\frac{2J_i}{N}+\frac{1}{N^2}\sum_{m=1}^n J_m^2\Big)\Big(\frac{1}{J_{i'}}-\frac{1}{N}\Big)
	+\Big(1-\frac{2J_{i'}}{N}+\frac{1}{N^2}\sum_{m=1}^n J_m^2\Big)\Big(\frac{1}{J_i}-\frac{1}{N}\Big)\\
	&\qquad \qquad \qquad \qquad \qquad \qquad -\frac{4}{N}\Big(\frac{1}{N^2}\sum_{m=1}^n J_m^2-\frac{J_i+J_{i'}}{N}\Big)
	\Bigg]\\
	&\quad
	+\mu_{4v}\Bigg[
	\frac{1}{N^4}\Bigg(\frac{(N-J_i)^2}{J_i}+\frac{(N-J_{i'})^2}{J_{i'}}+(N-J_i-J_{i'})\Bigg)
	\Bigg]\\
	&\quad
	+\mu_{2v}^2\Bigg(
	\Big(\frac{1}{J_i}-\frac{1}{N}\Big)\Big(\frac{1}{J_{i'}}-\frac{1}{N}\Big)
	+\frac{2}{N^2}
	-3\cdot \frac{1}{N^4}\Bigg(\frac{(N-J_i)^2}{J_i}+\frac{(N-J_{i'})^2}{J_{i'}}+(N-J_i-J_{i'})\Bigg)
	\Bigg).
\end{align*}

Summing the fourth-moment identity over $i$ and the cross-product identity over $i\neq i'$ gives a $2\times2$ system.
\[
\binom{\mu_{4u}}{\mu_{2u}^2}
=
\begin{pmatrix}
	a_{11}^{\obs} & a_{12}^{\obs}\\
	a_{21}^{\obs} & a_{22}^{\obs}
\end{pmatrix}^{-1}
\binom{
	\E\!\left[\sum_{i=1}^n (\bar y_i-\bar y^{\obs})^4 - T_4^{\obs}\right]
}{
	\E\!\left[\sum_{i\neq i'} (\bar y_i-\bar y^{\obs})^2(\bar y_{i'}-\bar y^{\obs})^2 - T_{22}^{\obs}\right]
},
\]
where

	\begin{align*}
		a_{11}^{\obs}
		&=\sum_{i=1}^n\Bigg[
		\Big(1-\frac{J_i}{N}\Big)^4+\frac{1}{N^4}\sum_{m\neq i}J_m^4
		\Bigg],\quad
		a_{12}^{\obs}
		=3\sum_{i=1}^n\Bigg(
		\Big[\Big(1-\frac{J_i}{N}\Big)^2+\frac{1}{N^2}\sum_{m\neq i}J_m^2\Big]^2
		-\Big[\Big(1-\frac{J_i}{N}\Big)^4+\frac{1}{N^4}\sum_{m\neq i}J_m^4\Big]
		\Bigg),\\[2pt]
		a_{21}^{\obs}
		&=\sum_{i\neq i'}\Bigg[
		\Big(\frac{J_i}{N}\Big)^2\Big(1-\frac{J_i}{N}\Big)^2
		+\Big(\frac{J_{i'}}{N}\Big)^2\Big(1-\frac{J_{i'}}{N}\Big)^2
		+\frac{1}{N^4}\sum_{m\neq i,i'}J_m^4
		\Bigg],\\[2pt]
		a_{22}^{\obs}
		&=\sum_{i\neq i'}\Bigg(
		\Big[1-\frac{2J_i}{N}+\frac{1}{N^2}\sum_{m=1}^n J_m^2\Big]
		\Big[1-\frac{2J_{i'}}{N}+\frac{1}{N^2}\sum_{m=1}^n J_m^2\Big]
		+2\Big[\frac{1}{N^2}\sum_{m=1}^n J_m^2-\frac{J_i+J_{i'}}{N}\Big]^2\\
		&\qquad\qquad
		-3\Bigg[
		\Big(\frac{J_i}{N}\Big)^2\Big(1-\frac{J_i}{N}\Big)^2
		+\Big(\frac{J_{i'}}{N}\Big)^2\Big(1-\frac{J_{i'}}{N}\Big)^2
		+\frac{1}{N^4}\sum_{m\neq i,i'}J_m^4
		\Bigg]
		\Bigg),\\[6pt]
		T_4^{\obs}
		&=
		6\mu_{2u}\mu_{2v}\sum_{i=1}^n
		\Big(1-\frac{2J_i}{N}+\frac{1}{N^2}\sum_{m=1}^n J_m^2\Big)\Big(\frac{1}{J_i}-\frac{1}{N}\Big)
		+\mu_{4v}\sum_{i=1}^n\Bigg[J_i\Big(\frac{1}{J_i}-\frac{1}{N}\Big)^4+(N-J_i)\Big(\frac{1}{N}\Big)^4\Bigg]\\
		&\quad
		+3\mu_{2v}^2\sum_{i=1}^n\Bigg[
		\Big(\frac{1}{J_i}-\frac{1}{N}\Big)^2
		-\Bigg(J_i\Big(\frac{1}{J_i}-\frac{1}{N}\Big)^4+(N-J_i)\Big(\frac{1}{N}\Big)^4\Bigg)
		\Bigg],\\[4pt]
		T_{22}^{\obs}
		&=
		\mu_{2u}\mu_{2v}\sum_{i\neq i'}
		\Bigg[
		\Big(1-\frac{2J_i}{N}+\frac{1}{N^2}\sum_{m=1}^n J_m^2\Big)\Big(\frac{1}{J_{i'}}-\frac{1}{N}\Big)
		+\Big(1-\frac{2J_{i'}}{N}+\frac{1}{N^2}\sum_{m=1}^n J_m^2\Big)\Big(\frac{1}{J_i}-\frac{1}{N}\Big)\\
		&\qquad\qquad\qquad\qquad
		-\frac{4}{N}\Big(\frac{1}{N^2}\sum_{m=1}^n J_m^2-\frac{J_i+J_{i'}}{N}\Big)
		\Bigg]
		+\mu_{4v}\sum_{i\neq i'}
		\frac{1}{N^4}\Bigg[
		\frac{(N-J_i)^2}{J_i}+\frac{(N-J_{i'})^2}{J_{i'}}+(N-J_i-J_{i'})
		\Bigg]\\
		&\quad
		+\mu_{2v}^2\sum_{i\neq i'}
		\Bigg[
		\Big(\frac{1}{J_i}-\frac{1}{N}\Big)\Big(\frac{1}{J_{i'}}-\frac{1}{N}\Big)
		+\frac{2}{N^2}
		-\frac{3}{N^4}\Bigg(
		\frac{(N-J_i)^2}{J_i}+\frac{(N-J_{i'})^2}{J_{i'}}+(N-J_i-J_{i'})
		\Bigg)
		\Bigg].
	\end{align*}

\section{Proofs for Three-Level Model}  \label{ap:3level}

Using within-$(i,j)$ variation,
\begin{align*}
	y_{ijk}-\bar y_{ij}
	&=w_{ijk}-\bar w_{ij}
	=\Big(1-\frac{1}{K_{ij}}\Big)w_{ijk}-\frac{1}{K_{ij}}\sum_{k'\neq k} w_{ijk'} .
\end{align*}
Applying Lemma~\ref{lma:powers} gives
\begin{align*}
	\E[(y_{ijk}-\bar y_{ij})^2]
	&=\mu_{2w}\Big[\Big(1-\frac{1}{K_{ij}}\Big)^2+(K_{ij}-1)\Big(\frac{1}{K_{ij}}\Big)^2\Big]
	=\mu_{2w}\frac{K_{ij}-1}{K_{ij}},\\
	\E[(y_{ijk}-\bar y_{ij})^3]
	&=\mu_{3w}\Big[\Big(1-\frac{1}{K_{ij}}\Big)^3+(K_{ij}-1)\Big(-\frac{1}{K_{ij}}\Big)^3\Big]
	=\mu_{3w}\frac{(K_{ij}-1)(K_{ij}-2)}{K_{ij}^2}.
\end{align*}

Using within-$i$ variation and group-averages across $j$,
\begin{align*}
	\bar y_{ij}-\bar y_i^{\grp}
	&=[v_{ij}-\bar v_i^{\grp}]+[\bar w_{ij}-\bar w_i^{\grp}]
	=\Big[\Big(1-\frac{1}{J_i}\Big)v_{ij}-\frac{1}{J_i}\sum_{j'\neq j} v_{ij'}\Big]
	+\Big[\Big(1-\frac{1}{J_i}\Big)\bar w_{ij}-\frac{1}{J_i}\sum_{j'\neq j}\bar w_{ij'}\Big],
\end{align*}
Applying Lemma~\ref{lma:powers} and using independence gives
\begin{align*}
	\E[(\bar y_{ij}-\bar y_i^{\grp})^2]
	&=\mu_{2v}\Big[\Big(1-\frac{1}{J_i}\Big)^2+(J_i-1)\Big(\frac{1}{J_i}\Big)^2\Big]
	 +\mu_{2w}\Big[\Big(1-\frac{1}{J_i}\Big)^2\frac{1}{K_{ij}}
	+\Big(\frac{1}{J_i}\Big)^2\sum_{j'\neq j}\frac{1}{K_{ij'}}\Big]\\
	&=\mu_{2v}\frac{J_i-1}{J_i}
	+\mu_{2w}\Big(\frac{1}{K_{ij}}-\frac{2}{J_iK_{ij}}+\frac{1}{J_i^2}\sum_{j'=1}^{J_i}\frac{1}{K_{ij'}}\Big),\\
	\E[(\bar y_{ij}-\bar y_i^{\grp})^3]
	&=\mu_{3v}\Big[\Big(1-\frac{1}{J_i}\Big)^3+(J_i-1)\Big(-\frac{1}{J_i}\Big)^3\Big]
	 +\mu_{3w}\Big[\Big(1-\frac{1}{J_i}\Big)^3\frac{1}{K_{ij}^2}
	+\Big(-\frac{1}{J_i}\Big)^3\sum_{j'\neq j}\frac{1}{K_{ij'}^2}\Big]\\
	&=\mu_{3v}\frac{(J_i-1)(J_i-2)}{J_i^2}
	+\mu_{3w}\Big(\frac{1}{K_{ij}^2}-\frac{3}{J_iK_{ij}^2}+\frac{3}{J_i^2K_{ij}^2}-\frac{1}{J_i^3}\sum_{j'=1}^{J_i}\frac{1}{K_{ij'}^2}\Big).
\end{align*}

Similarly, using within-$i$ variation and observation-averages across $j$,
\begin{align*}
	\bar y_{ij}-\bar y_i^{\obs}
	&=[v_{ij}-\bar v_i^{\obs}]+[\bar w_{ij}-\bar w_i^{\obs}]
	=\Big[\Big(1-\frac{K_{ij}}{K_i}\Big)v_{ij}-\frac{1}{K_i}\sum_{j'\neq j}K_{ij'}v_{ij'}\Big]
	+\Big[\Big(1-\frac{K_{ij}}{K_i}\Big)\bar w_{ij}-\frac{1}{K_i}\sum_{j'\neq j}K_{ij'}\bar w_{ij'}\Big].
\end{align*}
Applying Lemma~\ref{lma:powers} and using independence gives
\begin{align*}
	\E[(\bar y_{ij}-\bar y_i^{\obs})^2]
	&=\mu_{2v}\Big[\Big(1-\frac{K_{ij}}{K_i}\Big)^2+\frac{1}{K_i^2}\sum_{j'\neq j}K_{ij'}^2\Big]
	+\mu_{2w}\Big[\Big(1-\frac{K_{ij}}{K_i}\Big)^2\frac{1}{K_{ij}}
	+\frac{1}{K_i^2}\sum_{j'\neq j}\frac{K_{ij'}^2}{K_{ij'}}\Big]\\
	&=\mu_{2v}\Big(1-\frac{2K_{ij}}{K_i}+\frac{1}{K_i^2}\sum_{j'=1}^{J_i}K_{ij'}^2\Big)
	+\mu_{2w}\Big(\frac{1}{K_{ij}}-\frac{1}{K_i}\Big),\\
	\E[(\bar y_{ij}-\bar y_i^{\obs})^3]
	&=\mu_{3v}\Big[\Big(1-\frac{K_{ij}}{K_i}\Big)^3-\frac{1}{K_i^3}\sum_{j'\neq j}K_{ij'}^3\Big]
	 +\mu_{3w}\Big[\Big(1-\frac{K_{ij}}{K_i}\Big)^3\frac{1}{K_{ij}^2}
	-\frac{1}{K_i^3}\sum_{j'\neq j}\frac{K_{ij'}^3}{K_{ij'}^2}\Big]\\
	&=\mu_{3v}\Big(1-\frac{3K_{ij}}{K_i}+\frac{3K_{ij}^2}{K_i^2}-\frac{1}{K_i^3}\sum_{j'=1}^{J_i}K_{ij'}^3\Big)
	+\mu_{3w}\frac{(K_i-K_{ij})(K_i-2K_{ij})}{K_i^2K_{ij}^2}.
\end{align*}


Using between-$i$ variation and group-averages 
\begin{align*}
	\bar y_i^{\grp}-\bar y^{\grp}
	=[u_i-\bar u^{\grp}]+[\bar v_i^{\grp}-\bar v^{\grp}]+[\bar w_i^{\grp}-\bar w^{\grp}]
	&=\Big[\Big(1-\frac{1}{n}\Big)u_i-\frac{1}{n}\sum_{i'\neq i}u_{i'}\Big]
	+\Bigg[\Big(1-\frac{1}{n}\Big)\frac{1}{J_i}\sum_{j=1}^{J_i} v_{ij}
	-\frac{1}{n}\sum_{i'\neq i}\frac{1}{J_{i'}}\sum_{j=1}^{J_{i'}} v_{i'j}\Bigg]\\
	&\quad+\Bigg[\Big(1-\frac{1}{n}\Big)\frac{1}{J_i}\sum_{j=1}^{J_i}\frac{1}{K_{ij}}\sum_{k=1}^{K_{ij}} w_{ijk}
	-\frac{1}{n}\sum_{i'\neq i}\frac{1}{J_{i'}}\sum_{j=1}^{J_{i'}}\frac{1}{K_{i'j}}\sum_{k=1}^{K_{i'j}} w_{i'jk}\Bigg].
\end{align*}

Applying Lemma~\ref{lma:powers} and using independence gives
\begin{align*}
	\E[(\bar y_i^{\grp}-\bar y^{\grp})^2]
	&=\mu_{2u}\frac{n-1}{n}+\mu_{2v}\Big(\frac{1}{J_i}-\frac{2}{nJ_i}+\frac{1}{n^2}\sum_{i'=1}^n\frac{1}{J_{i'}}\Big)
	 +\mu_{2w}\Bigg(
	\Big(1-\frac{1}{n}\Big)^2\frac{1}{J_i^2}\sum_{j=1}^{J_i}\frac{1}{K_{ij}}
	+\Big(\frac{1}{n}\Big)^2\sum_{i'\neq i}\frac{1}{J_{i'}^2}\sum_{j=1}^{J_{i'}}\frac{1}{K_{i'j}}
	\Bigg),\\
	\E[(\bar y_i^{\grp}-\bar y^{\grp})^3]
	&=\mu_{3u}\frac{(n-1)(n-2)}{n^2}+\mu_{3v}\Big(\frac{1}{J_i^2}-\frac{3}{nJ_i^2}+\frac{3}{n^2J_i^2}-\frac{1}{n^3}\sum_{i'=1}^n\frac{1}{J_{i'}^2}\Big)\\
	&\quad+ \mu_{3w}\Bigg(
	\Big(1-\frac{1}{n}\Big)^3\frac{1}{J_i^3}\sum_{j=1}^{J_i}\frac{1}{K_{ij}^2}
	+\Big(-\frac{1}{n}\Big)^3\sum_{i'\neq i}\frac{1}{J_{i'}^3}\sum_{j=1}^{J_{i'}}\frac{1}{K_{i'j}^2}
	\Bigg).
\end{align*}

Similarly, using between-$i$ variation and observation-averages 
\begin{align*}
	\bar y_i^{\obs}-\bar y^{\obs}
	=[u_i-\bar u^{\obs}]+[\bar v_i^{\obs}-\bar v^{\obs}]+[\bar w_i^{\obs}-\bar w^{\obs}]
	&=\Big[\Big(1-\frac{K_i}{N}\Big)u_i-\frac{1}{N}\sum_{i'\neq i}K_{i'}u_{i'}\Big]
	+\Big[\Big(1-\frac{K_i}{N}\Big)\frac{1}{K_i}\sum_{j=1}^{J_i}K_{ij}v_{ij}
	-\frac{1}{N}\sum_{i'\neq i}\sum_{j=1}^{J_{i'}}K_{i'j}v_{i'j}\Big]\\
	&\quad +\Big[\Big(1-\frac{K_i}{N}\Big)\frac{1}{K_i}\sum_{j=1}^{J_i}\sum_{k=1}^{K_{ij}} w_{ijk}
	-\frac{1}{N}\sum_{i'\neq i}\sum_{j=1}^{J_{i'}}\sum_{k=1}^{K_{i'j}} w_{i'jk}\Big].
\end{align*}

Applying Lemma~\ref{lma:powers} and using independence gives
\begin{align*}
	\E[(\bar y_i^{\obs}-\bar y^{\obs})^2]
	&=\mu_{2u}\Big(1-\frac{2K_i}{N}+\frac{1}{N^2}\sum_{i'=1}^n K_{i'}^2\Big) 
	+\mu_{2v}\Bigg(
	\Big(1-\frac{K_i}{N}\Big)^2\frac{1}{K_i^2}\sum_{j=1}^{J_i}K_{ij}^2
	+\frac{1}{N^2}\sum_{i'\neq i}\sum_{j=1}^{J_{i'}}K_{i'j}^2
	\Bigg) +  \mu_{2w}\Big(\frac{1}{K_i}-\frac{1}{N}\Big),\\
	\E[(\bar y_i^{\obs}-\bar y^{\obs})^3]
	&=\mu_{3u}\Big(1-\frac{3K_i}{N}+\frac{3K_i^2}{N^2}-\frac{1}{N^3}\sum_{i'=1}^n K_{i'}^3\Big)
	+\mu_{3v}\Bigg(
	\Big(1-\frac{K_i}{N}\Big)^3\frac{1}{K_i^3}\sum_{j=1}^{J_i}K_{ij}^3
	+\Big(-\frac{1}{N}\Big)^3\sum_{i'\neq i}\sum_{j=1}^{J_{i'}}K_{i'j}^3
	\Bigg)\\
	&\quad +\mu_{3w}\frac{(N-K_i)(N-2K_i)}{N^2K_i^2}.
\end{align*}

\end{document}